\documentclass[final]{IEEEtran}

\usepackage{amsmath}
\usepackage{amsfonts}
\usepackage{array}
\usepackage{amsmath}
\usepackage{amssymb}
\usepackage{amsthm}
\usepackage{pifont}
\usepackage{textcomp}
\usepackage{stfloats}
\usepackage{url}
\usepackage{verbatim}
\usepackage{graphicx}
\usepackage{cite}
\usepackage{accents}
\usepackage{subfigure}
\usepackage{multirow}
\usepackage{color}
\usepackage{caption}
\usepackage{arydshln}
\usepackage[ruled,linesnumbered]{algorithm2e}

\newtheoremstyle{mytheoremstyle}{3pt}{3pt}{\itshape}{0cm}{\bfseries}{}{0.4em}{}

\theoremstyle{mytheoremstyle}
\newtheorem{assumption}{Assumption}
\newtheorem{proposition}{Proposition}

\begin{document}

\makeatletter
\newcommand\widebar[1]{\accentset{{\cc@style\underline{\mskip10mu}}}{#1}}
\def\wideubar{\underaccent{{\cc@style\underline{\mskip10mu}}}}
\makeatother

\makeatletter
\newcommand\widecheck[1]{\mathpalette\@widecheck{#1}}
\def\@widecheck#1#2{%
    \setbox\z@\hbox{\m@th$#1#2$}%
    \setbox\tw@\hbox{\m@th$#1%
       \widehat{%
          \vrule\@width\z@\@height\ht\z@
          \vrule\@height\z@\@width\wd\z@}$}%
    \dp\tw@-\ht\z@
    \@tempdima\ht\z@ \advance\@tempdima2\ht\tw@ \divide\@tempdima\thr@@
    \setbox\tw@\hbox{%
       \raise\@tempdima\hbox{\scalebox{1}[-1]{\lower\@tempdima\box\tw@}}}%
    {\ooalign{\box\tw@ \cr \box\z@}}}
\makeatother

\title{Vector Orthogonal Chirp Division Multiplexing Over Doubly Selective Channels}

\author{
Deyu Lu,~\IEEEmembership{Student Member,~IEEE}, Xiaoli Ma,~\IEEEmembership{Fellow,~IEEE}, and
Yiyin Wang,~\IEEEmembership{Senior Member,~IEEE}
\thanks{Deyu Lu and Yiyin Wang are with the State Key Laboratory of Submarine Geoscience, School of Automation and Sensing, Shanghai Jiao Tong University, Shanghai 200240, China.

Xiaoli Ma is with the School of Electrical and Computer Engineering,
Georgia Institute of Technology, Atlanta 30332, USA.}
}

\maketitle

\begin{abstract}
In this letter, we extend orthogonal chirp division multiplexing (OCDM) to vector OCDM (VOCDM) to provide more design freedom to deal with doubly selective channels. The VOCDM modulation is implemented by performing $M$ parallel $N$-size inverse discrete Fresnel transforms (IDFnT). Based on the complex exponential basis expansion model (CE-BEM) for doubly selective channels, we derive the VOCDM input-output relationship, and show performance tradeoffs of VOCDM with respect to (w.r.t.) its modulation parameters $M$ and $N$. Specifically, we investigate the diversity and peak-to-average power ratio (PAPR) of VOCDM w.r.t. $M$ and $N$. Under doubly selective channels, VOCDM exhibits superior diversity performance as long as the parameters $M$ and $N$ are configured to satisfy some constraints from the delay and the Doppler spreads of the channel, respectively. Furthermore, the PAPR of VOCDM signals decreases with a decreasing $N$. These theoretical findings are verified through numerical simulations.
\end{abstract}


\section{Introduction}

Orthogonal frequency division multiplexing (OFDM) plays an important role in wireless communications \cite{Wang2000}. It achieves great success for static multipath channels, but it is not the best fit for high mobility channels. Because OFDM is susceptible to Doppler shifts, which induce inter-carrier interference and degrade the system performance. Considering the limitations of OFDM in high mobility scenarios, novel modulations such as orthogonal time frequency space (OTFS) \cite{Hadani2017} have been developed for doubly selective channels. It has been proved in \cite{Deng2025} that OTFS is equivalent to vector OFDM (VOFDM) \cite{Xia2001} and orthogonal signal division multiplexing (OSDM) \cite{Suehiro2002}. Instead of modulating the symbols in the time-frequency domain as OFDM does, OTFS multiplexes the symbols in the delay-Doppler domain, where the channel representation is sparse \cite{Raviteja2018, Raviteja2019a}. Compared to OFDM, OTFS signals have a lower peak-to-average power ratio (PAPR) with a small size of Doppler grids \cite{Surabhi2019a}. Moreover, it has been proved that the asymptotic diversity order of the uncoded OTFS systems is one \cite{Surabhi2019b}. However, a high diversity order is observed for OTFS within the finite SNR region \cite{Surabhi2019b}.

Another novel modulation scheme called orthogonal chirp division multiplexing (OCDM) has been proposed in \cite{Ouyang2016}. Through inverse discrete Fresnel transform (IDFnT), OCDM modulates the transmitted symbols in the Fresnel domain. Thanks to the double-spreading property in the time and frequency domains, OCDM is robust to time burst interference (TBI) and narrowband interference (NBI) \cite{Omar2021}. The potential of the chirp waveform is also explored in affine frequency division multiplexing (AFDM) \cite{Bemani2021, Bemani2023}. Through discrete affine Fourier transform (DAFT), AFDM is able to adjust the parameters of chirp bases. Compared to OCDM, AFDM is more capable for doubly selective channels, as it enables maximum diversity. However, both OCDM and AFDM suffer from high PAPR issues and are lack of design flexibility to tradeoff performance, complexity, and PAPR.

In this letter, we extend OCDM to vector OCDM (VOCDM) to obtain more design freedom. Based on the complex-exponential basis expansion channel model (CE-BEM), we first establish the input-output relationship for VOCDM over doubly selective channels. Moreover, we show performance tradeoffs of a VOCDM system with respect to (w.r.t.) its modulation parameters $M$ and $N$. To be more specific, we analyze the diversity performance of VOCDM under doubly selective channels. We derive a diversity upper bound for VOCDM and provide the conditions w.r.t. $M$ and $N$ to maximize this bound. We also analyze the PAPR performance of VOCDM and show that the PAPR of VOCDM signals decreases with a reducing $N$. Numerical simulation results further corroborate our theoretical findings.

Notations: The entry in the $m$-th row and $n$-th column of the matrix ${\bf X}$ is denoted by $[{\bf X}]_{m,n}$, $m = 0, \dots, M-1$, $n = 0, \dots, N-1$. The vector ${\bf 0}_K$ (${\bf 1}_K$) denotes the all-zero (all-one) column vector of length $K$. The matrix ${\bf I}_K$ denotes the identity matrix of size $K$. The matrix ${\rm diag}({\bf x})$ denotes a diagonal matrix with ${\bf x}$ as its diagonal, while ${\rm circ}({\bf x})$ denotes a circulant matrix with ${\bf x}$ as its first column. The expression ${\rm mod}(a, b)$ denotes the modulo operation. The notation $|\cdot|$ counts the size of a set. Moreover, the operators $\otimes$, $(\cdot)^T$, $(\cdot)^{\mathcal H}$, $(\cdot)^{-1}$, $\lfloor \cdot \rfloor$, $\lceil \cdot \rceil$, $\| \cdot \|_k$, ${\mathcal R}(\cdot)$, ${\mathbb P}(\cdot)$, and ${\mathbb E}[ \cdot ]$ represent Kronecker product, transpose, conjugate transpose, inverse, floor, ceiling, $k$-norm, rank, probability, and expectation, respectively.

\section{System Model}

In this section, we establish the system model for VOCDM under doubly selective channels.

\subsection{VOCDM modulation}

Without loss of generality, we consider a single VOCDM block for transmission. The inter-block interference (IBI) is handled by the cyclic prefix (CP) added at the beginning of the transmitted block. The vector ${\bf s} = \left[ s_0, \dots, s_{K-1} \right]^T$ of length $K = MN$ denotes the transmitted data, where each data symbol $s_n$ is drawn from the constellation $\mathbb S$. The VOCDM modulation is given as follows:
\begin{equation} \label{VOCDM_modu}
    {\bf u} = \left( {\bf \Phi}_N^{\mathcal H} \otimes {\bf I}_M \right) {\bf s},
\end{equation}
where ${\bf u} = \left[ u_0, \dots, u_{K-1} \right]^T$ is the modulated vector, and ${\bf \Phi}_N$ is the size-$N$ DFnT matrix. Its element in the $m$-th row and $n$-th column is given by
\begin{equation}
    \left[ {\bf \Phi}_N \right]_{m,n} = \frac{1}{\sqrt{N}} e^{-j\frac{\pi}{4} + j\frac{\pi}{N} \left( m-n+\frac{{\rm mod}(N, 2)}{2} \right)^2}. \label{DFnT_mtx}
\end{equation}
Based on (\ref{VOCDM_modu}), the VOCDM modulation is a tradeoff between the single carrier (SC) and OCDM modulations.

\subsection{Doubly Selective Channel}

Doubly selective channels are characterized by time- and frequency-selectivity/dispersion. Based on the CE-BEM\cite{Visintin1996}, the doubly selective channel can be modeled as follows:
\begin{equation} \label{BEM_channel_model}
    h(l, i) = \sum_{q=-Q}^Q h_{l,q} e^{j\omega_q i}, \ l = 0, \dots, L,
\end{equation}
where $\omega_q = 2\pi q/K$, $L = \lfloor \tau_{\rm max} / T_s \rfloor$, and $Q = \lceil f_{\rm max} K T_s \rceil$ with $T_s$ being the sampling period. The maximum delay spread $\tau_{\rm max}$ and the maximum Doppler spread $f_{\rm max}$ are known and bounded. Moreover, the channel coefficients $h_{l,q}$, $l = 0, \dots, L$, $q = -Q, \dots, Q$ with the total number of $\rho = (L+1)(2Q+1)$ are invariant within a transmitted block.

Suppose that the timing and frequency synchronizations have been performed at the receiver. The received $k$-th sample is given by $r_k = \sum_{l=0}^L h(l, k) u_{{\rm mod}(k-l, K)} + v_k$, where $v_k$ denotes the noise. Collecting the samples $r_k$, $k = 0, \dots, K-1$ into a vector ${\bf r} = \left[ r_0, \dots, r_{K-1} \right ]^T$, we obtain
\begin{equation} \label{recei_VOCDM_signal_2}
    {\bf r} = {\bf H} {\bf u} + {\bf v},
\end{equation}
where the noise vector ${\bf v} = \left[ v_0, \dots, v_{K-1} \right]^T$. It follows zero-mean complex Gaussian distribution with the covariance matrix $\sigma^2 {\bf I}_K$, i.e., ${\bf v} \sim {\mathcal {CN}}({\bf 0}_K, \sigma^2 {\bf I}_K)$. Moreover, the channel matrix ${\bf H}$ is given by
\begin{equation} \label{BEM_channel_mtx}
    {\bf H} = \sum_{l=0}^{L} \sum_{q=-Q}^{Q} h_{l,q} {\bf D}_K^q {\bf \Pi}_K^l,
\end{equation}
where ${\bf D}_K = {\rm diag} \left( \left[ 1, e^{j\frac{2\pi}{K}}, \dots, e^{j\frac{2\pi}{K}(K-1)} \right]^T \right)$ denotes a size-$K$ diagonal matrix and ${\bf \Pi}_K = {\rm circ} \left( \left[ 0, 1, {\bf 0}_{K-2}^T \right]^T \right)$ denotes a size-$K$ permutation matrix.

\subsection{VOCDM demodulation}

Using (\ref{VOCDM_modu}) and (\ref{recei_VOCDM_signal_2}), we derive the demodulated VOCDM signal as follows:
\begin{equation}
    {\bf y} = ({\bf \Phi}_N \otimes {\bf I}_M) {\bf r} = {\bf H}_{\rm eff} {\bf s} + \bar {\bf v}, \label{IO_relation}
\end{equation}
where ${\bf H}_{\rm eff} = ({\bf \Phi}_N \otimes {\bf I}_M) {\bf H} \left( {\bf \Phi}_N^{\mathcal H} \otimes {\bf I}_M \right)$ is defined as the effective channel matrix for VOCDM, and $\bar {\bf v} = ({\bf \Phi}_N \otimes {\bf I}_M) {\bf v}$ is the transformed noise vector. With the derivation being provided in Appendix~\ref{append_eff_chn_mtx}, the effective channel matrix ${\bf H}_{\rm eff}$ is obtained as
\begin{equation} \label{eff_chn_mtx}
    {\bf H}_{\rm eff} = \sum_{l=0}^{L} \sum_{q=-Q}^{Q} \alpha_q h_{l,q} {\bf D}_K^q {\bf \Pi}_K^{l+qM},
\end{equation}
where $\alpha_q = e^{j\frac{\pi}{N}q({\rm mod}(N, 2)-q)}$. Based on (\ref{eff_chn_mtx}), the channel coefficient $h_{l,q}$ is carried on the $o_{l, q}$-th sub-diagonal of ${\bf H}_{\rm eff}$, i.e., $\left[ {\bf H}_{\rm eff} \right]_{{\rm mod} \left( k+o_{l, q}, K \right), k}$, $k = 0, \dots, K-1$, where $o_{l, q} = {\rm mod}(l+qM, K)$. Let us define the set
\begin{equation} \label{permu_orders}
    \mathbb O(L, Q, M, N) = \bigcup_{l=0}^L \bigcup_{q=-Q}^{Q} \{ o_{l, q} \},
\end{equation}
to collect the numbers $o_{l, q}$, $l = 0, \dots, L$, $q = -Q, \dots, Q$. Each entry of ${\bf H}_{\rm eff}$ contains at most one channel coefficient if and only if $|\mathbb O(L, Q, M, N)| = \rho$.

\begin{proposition} \label{prop_para_cond}
    When the VOCDM parameters $M$ and $N$ simultaneously satisfy $M \geq L+1$ and $N \geq 2Q+1$, we achieve $|\mathbb O(L, Q, M, N)| = \rho$.
\end{proposition}

\begin{proof}
    See Appendix~\ref{append_para_cond}.
\end{proof}

{\it Proposition~\ref{prop_para_cond}} establishes that with an appropriate selection of the modulation parameters $M$ and $N$, VOCDM guarantees each entry of ${\bf H}_{\rm eff}$ contains at most one channel coefficient. This structure of the effective channel matrix benefits the diversity order \cite{Bemani2023}. In the next section, we analyze how $M$ and $N$ affect the diversity performance of VOCDM in detail.

\section{Diversity Analysis}

Let us use the vector ${\bf h} = [h_{0,-Q}, \dots, h_{L,-Q}, \dots, h_{L,Q}]^T$ to collect $\rho$ channel coefficients. The following assumptions are considered for diversity analysis.

\begin{assumption} \label{assump_full_rank}
    The channel coefficients ${\bf h}$ follow zero-mean complex Gaussian distribution with a full rank covariance matrix ${\bf R}_h = \mathbb E[{\bf h} {\bf h}^{\mathcal H}]$, i.e., ${\bf h} \sim {\mathcal {CN}}({\bf 0}_{\rho}, {\bf R}_h)$.
\end{assumption}

\begin{assumption} \label{assump_ML}
    The channel coefficients ${\bf h}$ are known at the receiver, and the maximum likelihood (ML) detector is used for symbol detection.
\end{assumption}

Based on the input-output relationship (\ref{IO_relation}), the conditional pair-wise error probability (PEP) is upper bounded by \cite{Tarokh1998}
\begin{equation} \label{cond_PEP_1}
    {\mathbb P} \left( \left. {\bf s} \rightarrow {\bf s}^{\prime} \right| {\bf h} \right) \leq \exp \left( -\frac{\| {\bf H}_{\rm eff} {\bf e} \|_2^2}{4\sigma^2} \right).
\end{equation}
where ${\bf s}^{\prime} \in \mathbb S^{K \times 1}$ and ${\bf e} = {\bf s} - {\bf s}^{\prime}$ are the wrong detection vector and error vector, respectively. Using the expression of ${\bf H}_{\rm eff}$ in (\ref{eff_chn_mtx}), we have
\begin{equation} \label{error_dist}
    {\bf H}_{\rm eff} {\bf e} = {\bf C}({\bf s}, {\bf e}) {\bf h},
\end{equation}
where the matrix ${\bf C}({\bf s}, {\bf e})$ of size $K \times \rho$ is defined as
\begin{equation} \label{def_C}
    {\bf C}({\bf s}, {\bf e}) = [{\bf c}_{0, -Q}, \dots, {\bf c}_{L, -Q}, \dots, {\bf c}_{L, Q}],
\end{equation}
with the vectors ${\bf c}_{l,q}$ being given as
\begin{equation} \label{def_c}
    {\bf c}_{l,q} = \alpha_q {\bf D}_K^q {\bf \Pi}_K^{l+qM} {\bf e}.
\end{equation}
Let us decompose ${\bf R}_h = {\bf B}^{\mathcal H} {\bf B}$. Thus, the vector ${\bf h}$ can be decomposed as ${\bf h} = {\bf B} \bar {\bf h}$, where $\bar {\bf h}$ is the whitened Gaussian variable with zero-mean and covariance matrix ${\bf I}_{\rho}$. Plugging (\ref{error_dist}) and ${\bf h} = {\bf B} \bar {\bf h}$ into (\ref{cond_PEP_1}), we arrive at
\begin{equation} \label{cond_PEP_2}
    \!\! {\mathbb P} \left( \left. {\bf s} \rightarrow {\bf s}^{\prime} \right| \bar {\bf h} \right) \leq \exp \left( -\frac{\bar {\bf h}^{\mathcal H} {\bf B}^{\mathcal H} {\bf C}({\bf s}, {\bf e})^{\mathcal H} {\bf C}({\bf s}, {\bf e}) {\bf B} \bar {\bf h}}{4\sigma^2} \right).
\end{equation}
Taking the expectation of both sides of (\ref{cond_PEP_2}) w.r.t. $\bar {\bf h}$, we obtain an upper bound for the ergodic PEP ${\mathbb P}({\bf s} \rightarrow {\bf s}^{\prime})$:
\begin{align}
    \!\! {\mathbb P}({\bf s} \rightarrow {\bf s}^{\prime}) \leq \prod_{i=1}^{{\mathcal R}({\bf C}({\bf s}, {\bf e}))} \left( 1 \!+\! \frac{\lambda_i}{4\sigma^2} \right)^{-1} \leq \frac{(4\sigma^2)^{{\mathcal R}({\bf C}({\bf s}, {\bf e}))}}{\prod_{i=1}^{{\mathcal R}({\bf C}({\bf s}, {\bf e}))} \lambda_i}, \label{ave_PEP}
\end{align}
where $\lambda_i$, $i = 1, \dots, {\mathcal R}({\bf C}({\bf s}, {\bf e}))$ are non-zero eigenvalues of ${\bf B}^{\mathcal H} {\bf C}({\bf s}, {\bf e})^{\mathcal H} {\bf C}({\bf s}, {\bf e}) {\bf B}$. Based on (\ref{ave_PEP}), the diversity for VOCDM over doubly selective channels is defined as
\begin{equation} \nonumber
    G_d = \min_{{\bf s} \in \mathbb S^{K \times 1}} \min_{{\bf e} \ne {\bf 0}_K} {\mathcal R}({\bf C}({\bf s}, {\bf e})),
\end{equation}
Suppose that $K \geq \rho$, thus the maximum diversity provided by the doubly selective channel is $\rho$. A sufficient and necessary condition to enable the maximum diversity $\rho$ is that the matrix ${\bf C}({\bf s}, {\bf e})$ is full rank for all ${\bf e} \neq {\bf 0}_K$. However, VOCDM cannot enable the maximum diversity without precoding. It can be illustrated using an example. Given the error ${\bf e}_0 = \epsilon {\bf 1}_K$, where $\epsilon$ is a non-zero error, the vectors ${\bf c}_{0,0} = \dots = {\bf c}_{L,0} = \epsilon {\bf 1}_K$ are dependent on each other. As a result, the matrix ${\bf C}({\bf s}, {\bf e}_0)$ is rank deficient and the VOCDM diversity $G_d < \rho$. Moreover, we observe the following fact that a necessary condition to enable the maximum diversity is to make each entry of ${\bf H}_{\rm eff}$ contain at most one channel coefficient, i.e., $|\mathbb O(L, Q, M, N)|=\rho$. If the entry of ${\bf H}_{\rm eff}$ includes more than one channel coefficient, the incoherent summation of channel coefficients will cause channel fading and diversity loss. Although VOCDM cannot enable the full diversity, we can still design $M$ and $N$ to fulfill the aforementioned necessary condition based on {\it Proposition~\ref{prop_para_cond}}.

Furthermore, we emphasis that the appearance of ${\bf e}_0= \epsilon {\bf 1}_K$ is only possible for a few data realizations, e.g., ${\bf s} = s {\bf 1}_K$, where $s \in \mathbb S$. The probability of such data realizations is too low to make the error ${\bf e}_0$ dominate in the finite SNR region. It is similar to OTFS, where the diversity loss caused by the error ${\bf e}_0$ becomes observable only in the extremely high SNR region \cite{Surabhi2019b}. To capture the BER of VOCDM across the finite SNR region, we introduce another definition of diversity \cite{Mounir2009}:
\begin{equation} \label{diversity}
    G_d({\bf s}) = \min_{{\bf e} \ne {\bf 0}_K} {\mathcal R}({\bf C}({\bf s}, {\bf e})),
\end{equation}
which is the data-dependent diversity. In the finite SNR region, the BER of VOCDM is governed by the majority of $G_d({\bf s})$ instances \cite{Mounir2009}. The following proposition demonstrates that the data-dependent diversity of VOCDM is influenced by the choice of the parameters $M$ and $N$:
\begin{proposition} \label{prop_div_UB}
    For any ${\bf s} \in \mathbb S^{K \times 1}$, the data-dependent diversity of VOCDM is upper bounded by
    \begin{equation} \label{div_UB}
        G_d({\bf s}) \leq |\mathbb O(L, Q, M, N)|.
    \end{equation}
\end{proposition}

\begin{proof}
    See Appendix~\ref{append_div_UB}.
\end{proof}

The diversity upper bound $|\mathbb O(L, Q, M, N)|$ is useful for implementing VOCDM systems, because it links the channel parameters $L$, $Q$ and the VOCDM parameters $M$, $N$, and provides effective guidance for VOCDM parameter design. According to {\it Proposition~\ref{prop_para_cond}}, when $M \geq L+1$ and $N \geq 2Q+1$, the upper bound $|\mathbb O(L, Q, M, N)|$ reaches its maximum value $\rho$, i.e., $|\mathbb O(L, Q, M, N)| = \rho$. It is also the aforementioned necessary condition to enable the maximum diversity. As two special cases of VOCDM, both SC and OCDM fail to meet the aforementioned necessary condition. According to (\ref{permu_orders}), we have $\mathbb O(L, Q, 12, 1) = \{ 0, \dots, L \}$ for SC and $\mathbb O(L, Q, 1, 12) = \{ 0, \dots, L+Q \} \cup \{ K-Q, \dots, K-1 \}$ for OCDM. Therefore, the diversity upper bound $|\mathbb O(L, Q, M, N)|$ for SC and OCDM is given by $L+1$ and $L+2Q+1$, respectively, which are both strictly lower than $\rho$ when $L > 0$ and $Q > 0$. It explains why SC and OCDM are sub-optimal under doubly selective channels.

\begin{figure}
    \centering
    \includegraphics[width=1\linewidth]{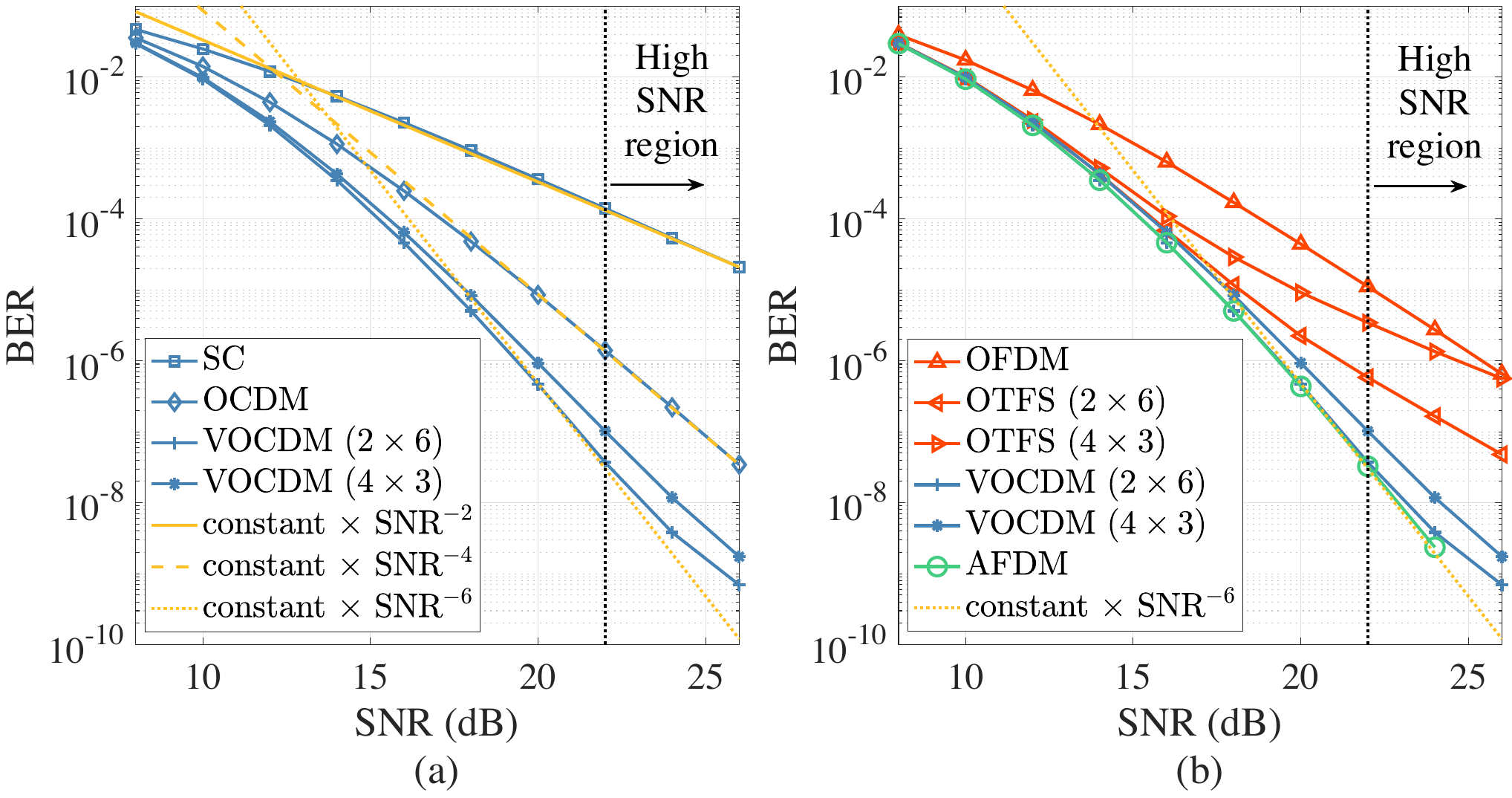}
    \caption{BER of VOCDM and its competitors under doubly selective channels with $L = 1$ and $Q = 1$. The ML detector is used.}
    \label{fig_BER}
\end{figure}

Fig.~\ref{fig_BER}(a) presents the BER performance of VOCDM with different parameter settings. Due to the high complexity of the ML detector, the block size is set as $K = 12$ and the QPSK constellation is adopted. The SNR is defined as $1/\sigma^2$. The tested channel takes independent identically distributed (i.i.d.) coefficients with $L = 1$ and $Q = 1$. Therefore, the maximum diversity provided by the channel is $6$. For SC, OCDM, and two VOCDM schemes with the parameters $(M,N) = (12,1)$, $(1,12)$, $(2, 6)$, and $(4, 3)$, respectively, their data-dependent diversity $G_d({\bf s})$ is upper bounded by $|\mathbb O(L, Q, M, N)| = 2$, $4$, $6$, and $6$, accordingly. As shown in Fig.~\ref{fig_BER}(a), the observed diversity of the tested VOCDM schemes almost approach their respective upper bound in the finite SNR region. It implies that $G_d({\bf s}) = |\mathbb O(L, Q, M, N)|$ for most data realizations. The VOCDM schemes with parameters $(2, 6)$ and $(4, 3)$ cannot achieve the maximum diversity in the high SNR region, because VOCDM cannot enable the maximum diversity under doubly selective channels, i.e., $G_d < \rho$.

Fig.~\ref{fig_BER}(b) compares the BER performance of VOCDM with other modulation schemes. Similar to VOCDM, the BER performance of OTFS depends on its modulation parameters $M$ and $N$, which determine the sizes of the delay and Doppler domain grids, respectively. Among all tested OTFS schemes, the best configuration is given by $(M, N) = (2, 6)$, but it performs worse than the VOCDM schemes with parameters $(2, 6)$ and $(4, 3)$. We also test the BER performance of AFDM \cite{Bemani2023}. With the appropriate modulation parameters, AFDM achieves the maximum diversity in the high SNR region.

\section{PAPR Analysis}

For simplicity, we consider the normalized constellation set $\mathbb S$ for PAPR analysis. Moreover, we assume that the data vector ${\bf s}$ is randomly generated from $\mathbb S^{K \times 1}$. The covariance matrix for ${\bf s}$ is given by ${\bf I}_K$. Since the VOCDM modulation matrix ${\bf \Phi}_N^{\mathcal H} \otimes {\bf I}_M$ is unitary, the covariance matrix for ${\bf u}$ is also given by ${\bf I}_K$. Therefore, the instantaneous PAPR of ${\bf u}$ is given by
\begin{equation} \nonumber
    {\rm PAPR}({\bf s}) = \frac{\| {\bf u} \|_{\infty}^2}{\mathbb E \left[ \| {\bf u} \|_2^2 \right] / K} = \| {\bf u} \|_{\infty}^2.
\end{equation}
Note that the vector ${\bf u}$ is composed of the $M$ sub-vectors $\bar {\bf u}_m = {\bf \Phi}_N^{\mathcal H} \bar {\bf s}_m$, $m = 0, \dots, M-1$, where $\left[ \bar {\bf u}_m \right]_n = [{\bf u}]_{nM+m}$ and $\left[ \bar {\bf s}_m \right]_n = [{\bf s}]_{nM+m}$. As a result, the instantaneous PAPR of ${\bf u}$ is the maximum one of $\bar {\bf u}_m$, $m = 0, \dots, M-1$:
\begin{equation} \nonumber
    {\rm PAPR}({\bf s}) = \max_{m \in \mathbb M} \left \| \bar {\bf u}_m \right \|_{\infty}^2 = \max_{m \in \mathbb M} \left \| {\bf \Phi}_N^{\mathcal H} \bar {\bf s}_m \right \|_{\infty}^2,
\end{equation}
where $\mathbb M = \{ 0, \dots, M-1 \}$. When $N$ is large enough, the complementary cumulative distribution function (CCDF) of $\left \| {\bf \Phi}_N^{\mathcal H} \bar {\bf s}_m \right \|_{\infty}^2$ can be approximated by ${\mathbb P} \left( \left \| {\bf \Phi}_N^{\mathcal H} \bar {\bf s}_m \right \|_{\infty}^2 > \gamma \right) \approx 1 - (1 - e^{-\gamma})^N$ \cite{Omar2021}. Since $\left \| {\bf \Phi}_N^{\mathcal H} \bar {\bf s}_m \right \|_{\infty}^2$, $m = 0, \dots, M-1$ are i.i.d. variables, the CCDF of ${\rm PAPR}({\bf s})$ is approximated by
\begin{align} \nonumber
    {\mathbb P} \left( {\rm PAPR}({\bf s}) > \gamma \right) & = 1 - {\mathbb P} \left( \left \| {\bf \Phi}_N^{\mathcal H} \bar {\bf s}_0 \right \|_{\infty}^2 \leq \gamma \right)^M \nonumber \\
    & \approx 1 - (1 - e^{-\gamma})^K, \label{PAPR_CCDF}
\end{align}
which is coincident with the CCDF of the instantaneous PAPR of the multicarrier signals of length $K$ \cite{Wang2004, Omar2021}.

When $N$ is selected to be small numbers, the approximation in (\ref{PAPR_CCDF}) is not accurate. Under such scenarios, the overall PAPR is more appropriate. The overall PAPR for VOCDM is derived by maximizing the instantaneous PAPR over all data realizations:
\begin{equation} \label{PAPR_all}
    {\rm PAPR}_{\rm all} = \max_{{\bf s} \in \mathbb S^{K \times 1}} {\rm PAPR}({\bf s}) = \max_{\bar {\bf s}_0 \in \mathbb S^{N \times 1}} \left \| {\bf \Phi}_N^{\mathcal H} \bar {\bf s}_0 \right \|_{\infty}^2.
\end{equation}
It is hard to derive the closed-form of ${\rm PAPR}_{\rm all}$ from (\ref{PAPR_all}). Using the inequality $\left \| {\bf \Phi}_N^{\mathcal H} \bar {\bf s}_0 \right \|_{\infty} \leq \left \| {\bf \Phi}_N^{\mathcal H} \bar {\bf s}_0 \right \|_2$, we derive an upper bound for ${\rm PAPR}_{\rm all}$:
 \begin{equation} \label{PAPR_all_UB}
    {\rm PAPR}_{\rm all} \leq a N,
\end{equation}
where $a = \max_{s \in \mathbb S} \| s \|_2^2$. Based on the upper bound (\ref{PAPR_all_UB}), the PAPR of VOCDM signals decreases with a decreasing $N$. On the other hand, the overall PAPR for OTFS is given by
\begin{equation} \label{PAPR_all_OTFS}
    {\rm PAPR}_{\rm all}^{\rm OTFS} = \max_{\bar {\bf s}_0 \in \mathbb S^{N \times 1}} \left \| {\bf F}_N^{\mathcal H} \bar {\bf s}_0 \right \|_{\infty}^2 = aN,
\end{equation}
where the maximum is achieved by $\bar {\bf s}_0 = b {\bf 1}_N$ with $\| b \|_2^2 = a$. Based on \eqref{PAPR_all_UB}
and \eqref{PAPR_all_OTFS}, we have $ {\rm PAPR}_{\rm all} \le  {\rm PAPR}_{\rm all}^{\rm OTFS}$. In Table~\ref{tab_PAPR_all}, we provide some examples of overall PAPRs of VOCDM and OTFS via exhaustive search. Compared to OTFS, VOCDM signals have a smaller overall PAPR.

\begin{table}
    \centering
    \caption{The overall PAPR of VOCDM and OTFS with different $N$'s and constellations.}
    \begin{tabular}{cccccc}
        \hline
        {\bf VOCDM}/OTFS & $N = 3$ & $N = 5$ & $N = 9$ & $N = 12$ \\
        \hline
        BPSK & ${\bf 2.33}/3$ & ${\bf 3.59}/5$ & ${\bf 5.28}/9$ & ${\bf 7.97}/12$ \\
        \hline
        QPSK & ${\bf 2.82}/3$ & ${\bf 4.44}/5$ & ${\bf 7.90}/9$ & ${\bf 10.09}/12$\\
        \hline
        $4$-PAM & ${\bf 4.2}/5.4$ & ${\bf 6.46}/9$ & ${\bf 9.51}/16.2$ & ${\bf 14.34}/21.6$ \\
        \hline
    \end{tabular}
    \label{tab_PAPR_all}
\end{table}

Fig.~\ref{fig_PAPR_CCDF} presents a comparison of instantaneous PAPR among different modulation schemes, where the BPSK constellation is used and the block size is set as $K = 400$. As shown in Fig.~\ref{fig_PAPR_CCDF}, the CCDF of the instantaneous PAPR of VOCDM, OTFS, OFDM, OCDM and AFDM with a large parameter $N$ is close to the theoretical CCDF, i.e., $1 - (1 - e^{-\gamma})^{400}$. However, for VOCDM and OTFS with small $N$'s, their CCDFs are significantly lower than the theoretical CCDF. Because the instantaneous PAPR is upper bounded by the overall PAPR, which tends to decrease as $N$ decreases. For both large and small $N$'s, VOCDM exhibits superior PAPR performance compared to other modulation schemes, i.e., OFDM, OTFS, OCDM and AFDM.

\begin{figure}
    \centering
    \includegraphics[width=0.85\linewidth]{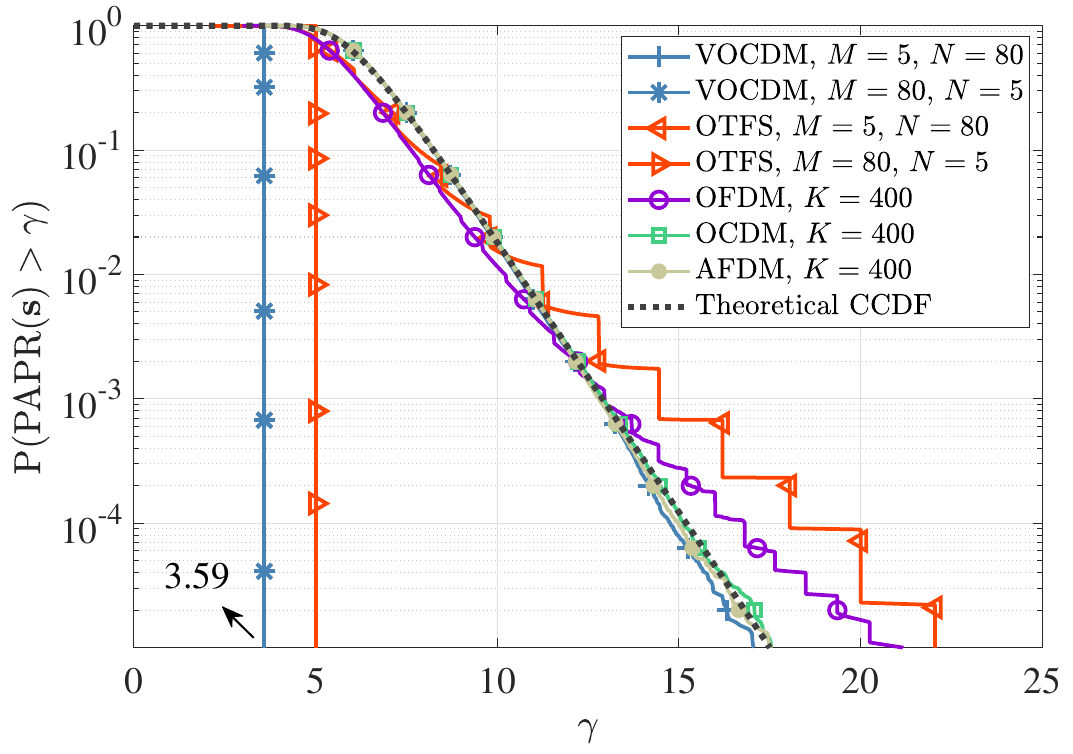}
    \caption{Comparison of instantaneous PAPR across different modulation schemes.}
    \label{fig_PAPR_CCDF}
\end{figure}

\section{Conclusions}

In this letter, we extend OCDM to VOCDM to achieve more design freedom for doubly selective channels, and analyze its performance tradeoffs w.r.t. $M$ and $N$. A diversity upper bound $|\mathbb O(L, Q, M, N)|$ is derived for VOCDM. It links the channel parameters $L$, $Q$ and the VOCDM parameters \mbox{$M$, $N$}. The proposed diversity upper bound is maximized when $M \geq L+1$ and $N \geq 2Q+1$. The instantaneous PAPR and overall PAPR of VOCDM are analyzed, respectively. The PAPR of VOCDM decreases with a decreasing $N$. Simulation results verify our theoretical findings.

\begin{appendices}

\section{Derivation of the effective channel matrix ${\bf H}_{\rm eff}$} \label{append_eff_chn_mtx}

Using (\ref{DFnT_mtx}), the IDFnT matrix ${\bf \Phi}_N^{\mathcal H}$ is circulant and can be written as ${\bf \Phi}_N^{\mathcal H} = {\rm circ}({\boldsymbol \phi}_0) = \left[ {\boldsymbol \phi}_0, {\bf \Pi}_N {\boldsymbol \phi}_0, \dots, {\bf \Pi}_N^{N-1} {\boldsymbol \phi}_0 \right]$, where ${\boldsymbol \phi}_0$ denotes the first column of ${\bf \Phi}_N^{\mathcal H}$. Consequently, ${\bf \Phi}_N^{\mathcal H} \otimes {\bf I}_M = \left[ {\boldsymbol \phi}_0 \otimes {\bf I}_M, {\bf \Pi}_N {\boldsymbol \phi}_0 \otimes {\bf I}_M, \dots, {\bf \Pi}_N^{N-1} {\boldsymbol \phi}_0 \otimes {\bf I}_M \right] = \left[ {\boldsymbol \phi}_0 \otimes {\bf t}, {\bf \Pi}_K ({\boldsymbol \phi}_0 \otimes {\bf t}), \dots, {\bf \Pi}_K^{K-1} ({\boldsymbol \phi}_0 \otimes {\bf t}) \right] = {\rm circ}({\boldsymbol \phi}_0 \otimes {\bf t})$ is also a circulant matrix, where ${\bf t} = \left[ 1, {\bf 0}_{M-1}^T \right]^T$. Using (\ref{BEM_channel_mtx}) and the commutativity between the circulant matrices ${\bf \Pi}_K^l$ and ${\bf \Phi}_N^{\mathcal H} \otimes {\bf I}_M$, the effective channel matrix ${\bf H}_{\rm eff}$ is obtained as
\begin{align}
    \! {\bf H}_{\rm eff} & = ({\bf \Phi}_N \otimes {\bf I}_M) {\bf H} \left( {\bf \Phi}_N^{\mathcal H} \otimes {\bf I}_M \right) \nonumber \\
    & = \sum_{l=0}^{L} \sum_{q=-Q}^{Q} h_{l,q} ({\bf \Phi}_N \otimes {\bf I}_M) {\bf D}_K^q \left( {\bf \Phi}_N^{\mathcal H} \otimes {\bf I}_M \right) {\bf \Pi}_K^l. \label{eff_chn_mtx_derive_1}
\end{align}
Note that the diagonal matrix ${\bf D}_K$ can be written as ${\bf D}_K = {\bf D}_N \otimes {\bf \Lambda}_M$, where ${\bf \Lambda}_M = [ {\bf D}_K ]_{0:M-1, 0:M-1}$. It leads us to
\begin{align}
    ({\bf \Phi}_N \otimes {\bf I}_M) {\bf D}_K^q \left( {\bf \Phi}_N^{\mathcal H} \otimes {\bf I}_M \right) & = {\bf \Phi}_N {\bf D}_N^q {\bf \Phi}_N^{\mathcal H} \otimes {\bf \Lambda}_M^q \nonumber \\
    & = \alpha_q {\bf D}_N^q {\bf \Pi}_N^q \otimes {\bf \Lambda}_M^q \nonumber \\
    & = \alpha_q ({\bf D}_N^q \otimes {\bf \Lambda}_M^q) ({\bf \Pi}_N^q \otimes {\bf I}_M) \nonumber \\
    & = \alpha_q {\bf D}_K^q {\bf \Pi}_K^{qM},  \label{eff_chn_mtx_derive_2}
\end{align}
where the term ${\bf \Phi}_N {\bf D}_N^q {\bf \Phi}_N^{\mathcal H} = \alpha_q {\bf D}_N^q {\bf \Pi}_N^q$ with $\alpha_q = e^{j\frac{\pi}{N}q({\rm mod}(N, 2)-q)}$ can be derived in an element-wise manner \cite{Wang2023}. Plugging (\ref{eff_chn_mtx_derive_2}) into (\ref{eff_chn_mtx_derive_1}), the effective channel matrix ${\bf H}_{\rm eff}$ is obtained in (\ref{eff_chn_mtx}).

\section{Proof of {\it Proposition~\ref{prop_para_cond}}} \label{append_para_cond}

\begin{figure*}
    \centering
    \begin{align}
       & \underbrace{ 0, \dots, L}_{\mathbb D_0} < \dots < \underbrace{QM, \dots, L+QM}_{\mathbb D_Q} < \underbrace{MN-QM, \dots, MN+L-QM}_{\mathbb D_{-Q}} < \dots < \underbrace{MN-M, \dots, MN+L-M}_{\mathbb D_{-1}}. \label{permu_orders_relation} \\
       & \rule{\textwidth}{0.4pt} \nonumber
    \end{align}
\end{figure*}

We begin by decomposing the set defined in (\ref{permu_orders}) as $\mathbb O(L, Q, M, N) = \bigcup_{q=-Q}^{Q} \mathbb O_q$, where $\mathbb O_q = \bigcup_{l=0}^L o_{l, q} = \bigcup_{l=0}^L \{ {\rm mod}(l+qM, MN) \}$. Using the inequalities $M \geq L+1$ and $N \geq 2Q+1$, we have i) $0 \leq l+qM < MN$, $\forall l \in \{ 0, \dots, L \}$, $\forall q \in \{ 0, \dots, Q \}$; and ii) $-MN < l+qM < 0$, $\forall l \in \{ 0, \dots, L \}$, $\forall q \in \{ -Q, \dots, -1 \}$. Therefore, the subset $\mathbb O_q = \{ qM, \dots, L+qM \}$, $\forall q \in \{ 0, \dots, Q \}$ and $\mathbb O_q = \{ MN+qM, \dots, MN+L+qM \}$, $\forall q \in \{ -Q, \dots, -1 \}$. It is not difficult to find that $|\mathbb O_q| = L+1$, $\forall q \in \{ -Q, \dots, Q \}$. Moreover, using the inequalities $M \geq L+1$ and $N \geq 2Q+1$ again, we establish the relationship (\ref{permu_orders_relation}) as listed on the top of this page, which implies that $\mathbb O_{q_1} \cap \mathbb O_{q_2} = \emptyset$, $\forall q_1 \ne q_2$. It leads us to $|\mathbb O(L, Q, M, N)| = \left| \bigcup_{q=-Q}^Q \mathbb O_q \right| = \bigcup_{q=-Q}^Q |\mathbb O_q| = \rho$.

\section{Proof of {\it Proposition~\ref{prop_div_UB}}} \label{append_div_UB}

Let us recall the definition of ${\bf C}({\bf s}, {\bf e})$ in (\ref{def_C}), and reorganize the columns of ${\bf C}({\bf s}, {\bf e})$ as follows:
\begin{equation} \label{def_C_tau}
     \!\!\! {\bf C}_{\tau}({\bf s}, {\bf e}) = \left[ {\bf c}_{l_1^{\tau}, q_1^{\tau}}, \dots, {\bf c}_{l_{B_{\tau}}^{\tau}, q_{B_{\tau}}^{\tau}} \right], \tau \in \mathbb O(L, Q, M, N),
\end{equation}
where the pairs $(l_1^{\tau}, q_1^{\tau}), \dots, (l_{B_{\tau}}^{\tau}, q_{B_{\tau}}^{\tau})$ are drawn from the set
\begin{equation} \nonumber
    \mathbb B_{\tau} = \{ (l, q) : l \in \{ 0, \dots, L \}, \, q \in \{ -Q, \dots, Q \}, \  o_{l, q} = \tau \} \nonumber
\end{equation}
with an arbitrary order, and the number $B_{\tau} = |\mathbb B_{\tau}|$. Recalling the definition of ${\bf c}_{l, q}$ in (\ref{def_c}), and plugging it into (\ref{def_C_tau}), we obtain ${\bf C}_{\tau}({\bf s}, {\bf e}) = \left[ \alpha_{q_1^{\tau}} {\bf D}_K^{q_1^{\tau}} {\bf \Pi}_K^{\tau} {\bf e}, \dots, \alpha_{q_{B_{\tau}}^{\tau}} {\bf D}_K^{q_{B_{\tau}}^{\tau}} {\bf \Pi}_K^{\tau} {\bf e} \right]$. For any realization of ${\bf s} \in \mathbb S^{K \times 1}$ with any constellation $\mathbb S$, there always exists a pairwise error ${\bf e}_1 = [\epsilon, {\bf 0}_{K-1}^T]^T$, where $\epsilon$ is a non-zero number. When $B_{\tau} > 1$, the vectors $\alpha_{q_i^{\tau}} {\bf D}_K^{q_i^{\tau}} {\bf \Pi}_K^{\tau} {\bf e}_1 = \left[ {\bf 0}_{\tau}^T, \alpha_{q_i^{\tau}} e^{j\frac{2\pi}{K} q_i^{\tau} \tau} \epsilon, {\bf 0}_{K-\tau-1}^T \right]^T$, $i = 1, \dots, B_{\tau}$ are mutually dependent. Given the error ${\bf e}_1$, we have ${\mathcal R}({\bf C}_{\tau}({\bf s}, {\bf e}_1)) = 1$, $\forall \tau \in \mathbb O(L, Q, M, N)$. Note that the columns of the matrix ${\bf C}({\bf s}, {\bf e})$ are composed of the columns of the matrices ${\bf C}_{\tau}({\bf s}, {\bf e})$, $\tau \in \mathbb O(L, Q, M, N)$. Thus, the data-dependent diversity can be upper bounded as follows:
\begin{align}
    G_d({\bf s}) & \leq \min_{{\bf e} \ne {\bf 0}_K} \sum_{\tau \in \mathbb O(L, Q, M, N)} {\mathcal R}({\bf C}_{\tau}({\bf s}, {\bf e})) \nonumber \\
    & \leq \sum_{\tau \in \mathbb O(L, Q, M, N)} {\mathcal R}({\bf C}_{\tau}({\bf s}, {\bf e}_1)) \nonumber \\
    & = |\mathbb O(L, Q, M, N)|. \nonumber
\end{align}

\end{appendices}

\bibliographystyle{IEEEtran}
\bibliography{cite}

\end{document}